\documentclass[conference]{IEEEtran}
\usepackage{graphicx}
\usepackage{amsmath}
\usepackage{amssymb}
\usepackage{amsthm}         % proof 环境
\newtheorem{lemma}{Lemma}   % lemma 环境
\usepackage{algorithm}      
\usepackage{algorithmicx}   
\usepackage{algpseudocode} 
\usepackage{booktabs} 
\usepackage{cite}
\usepackage{balance} 
\usepackage{multirow}
\usepackage[caption=false]{subfig}

\ifCLASSINFOpdf
\else
\fi
\begin{document}
\title{TGL-APT: Temporal Graph Learning with Graph Distillation for Efficient APT Investigation}

%\author{}
\author{\IEEEauthorblockN{Jing Chen}
	\IEEEauthorblockA{Fujian Normal University\\
		qbx20240213@yjs.fjnu.edu.cn}
	\and
	\IEEEauthorblockN{Ayong Ye\textsuperscript{*}} 
	\IEEEauthorblockA{
		Fujian Normal University\\ 
		yay@fjnu.edu.cn\\
		\textsuperscript{*}Corresponding author
	}
	\and
	\IEEEauthorblockN{Yuanhuang Liu}
	\IEEEauthorblockA{Fujian Normal University\\
		qsx20231394@student.fjnu.edu.cn}
	\and
	\IEEEauthorblockN{Yuexin Zhang}
	\IEEEauthorblockA{Fujian Normal University\\
		yxzhang@fjnu.edu.cn}}


		%Telephone: (800) 555--1212, Fax: (888) 555--1212}


\maketitle

\begin{abstract}
	Advanced Persistent Threat (APT) attacks pose a critical challenge to modern systems, as their stealthy, multi-stage nature renders conventional detection methods ineffective. While provenance graphs provide rich behavioral context for attack investigation, attack-relevant evidence is often sparse and embedded in large volumes of routine system activity, making full-graph learning both computationally expensive and difficult to correlate over long attack sequences. We present TGL-APT, an adaptive investigation framework built on the observation that attack-relevant information is non-uniformly distributed and often mediated by structurally influential or behaviorally distinctive entities, which we characterize as information-bottleneck nodes. TGL-APT combines three complementary components: (1) information-bottleneck-guided graph distillation that suppresses provenance redundancy while bounding structural distortion and preserving causal reachability; (2) adaptive temporal graph learning that continuously refines the core node set as node relevance evolves; and (3) cross-spatiotemporal attack fingerprint alignment that associates fragmented suspicious activities across different entities and time windows. Finally, causal expansion and stage characterization reconstruct coherent attack processes for investigation. Experiments on three DARPA E3 datasets show F1-scores of 95.7\%, 90.9\%, and 88.9\%, while reducing training time, detection latency, and memory usage by approximately 39\%, 33\%, and 22\%, respectively, compared with KAIROS. These results demonstrate that TGL-APT effectively balances detection performance, computational efficiency, and investigation capability for provenance-based APT analysis.
\end{abstract}


% \fi
\IEEEpeerreviewmaketitle

\section{Introduction}
\label{sec:intro}

Advanced Persistent Threat (APT) attacks, with the characteristics of stealthiness, long-term latency, and multi-stage infiltration, have become one of the most serious threats to critical infrastructure and information systems \cite{Zipperle2022}. Traditional signature-based detection tools can hardly cope with the low-frequency, slow-changing, and highly concealed behaviors of APT attacks. System provenance graphs constructed from audit logs have become the mainstream paradigm for APT traceability due to their ability to record entity dependencies and behavioral processes \cite{li2021threat}.

Recent provenance-based approaches have substantially advanced APT detection by modeling structural and temporal dependencies in system activities. In particular, graph-learning-based methods learn representations of processes, files, and network interactions from provenance graphs, while temporal approaches further capture the evolution of system behavior for anomaly detection and attack investigation \cite{Magic2024,Flash2024,Kairos2024}. These methods demonstrate the effectiveness of provenance-based temporal modeling for detecting stealthy attacks that are difficult to identify from isolated events.

However, the increasing scale of provenance data creates a fundamental signal-to-noise dilemma for APT detection. In particular, provenance-based APT analysis exhibits a \emph{data paradox}: collecting more system activities provides richer behavioral context, yet attack-relevant evidence remains sparse, causing its effective signal-to-noise ratio to decrease as routine provenance accumulates. Long-running systems may generate millions of nodes and edges, while attack-related evidence is easily obscured by repetitive benign activities. Consequently, full-graph learning not only incurs substantial computational overhead \cite{Magic2024,Flash2024,Kairos2024}, but also processes large amounts of information with limited security relevance. Meanwhile, indiscriminate graph reduction may discard rare but critical entities or break causal dependencies \cite{Li2025TeRed,VGAE4PG2026}. These limitations raise a more fundamental question: what information in a provenance graph is actually necessary for preserving attack-relevant evidence?

Existing full-graph learning and generic graph reduction do not explicitly answer this question. In provenance graphs, attack evidence is often sparse and non-uniformly distributed, tending to concentrate around a limited set of structurally influential or behaviorally distinctive entities. Moreover, as an APT campaign evolves, the entities carrying such evidence may change over time, causing related activities to appear across different processes, files, network endpoints, and time windows. Such temporal relevance drift makes a statically distilled node set increasingly mismatched to the evolving attack context. This temporal dispersion makes entity co-occurrence insufficient for correlating related activities and may fragment a continuous attack into disconnected anomaly sequences \cite{Zhang2025TAPAS,wang2025pathwatcher,Kairos2024,R-CAID2024}. Even when suspicious events are detected, isolated anomalies alone cannot reveal their causal dependencies or progression across attack stages \cite{Dong2023,jiang2025orthrus}. Therefore, the key challenge is to understand where attack-relevant information is concentrated, how its importance evolves over time, and how distributed evidence can be connected into a coherent attack process.

To address this question, we formulate provenance distillation from an information-bottleneck perspective and propose TGL-APT, a temporal graph learning framework that progressively extracts, tracks, and associates attack-relevant information from large-scale provenance data. The framework first employs information-bottleneck-guided graph distillation to identify structurally and behaviorally informative entities and preserve their critical causal context while suppressing redundant system activities. It then uses an adaptive focus mechanism to dynamically refine the core node set according to model-derived signals, enabling the analysis focus to follow changes in information relevance over time. To recover attack evidence fragmented across different entities and time windows, cross-spatiotemporal fingerprint alignment combines entity-level and semantic correlations to connect related suspicious activities. Finally, causal expansion and stage characterization reconstruct the correlated evidence into coherent attack chains and identify attack stages, transforming fragmented anomalies into interpretable investigation results. The main contributions are summarized as follows:

(1) We formulate provenance-based APT investigation from an Information Bottleneck perspective and formally characterize IB nodes from information-theoretic, graph-topological, and temporal perspectives. We further establish a structural-causal preservation and downstream stability bound, providing theoretical support for preserving attack-relevant information during graph distillation.

(2) We propose TGL-APT, an efficient and adaptive framework that addresses the limitations of static graph reduction. Its Adaptive Focus Mechanism dynamically updates critical nodes using model-derived attention weights and embedding deviations, enabling the analysis focus to follow evolving attack-relevant information over time.

(3) To address the ``synonymous but heterogeneous'' nature of attack behaviors across different entities and time windows, we design a cross-spatiotemporal attack fingerprint alignment method. By combining TF-IDF-weighted fingerprints with contrastive learning, TGL-APT associates semantically related suspicious activities across temporal windows without requiring direct entity co-occurrence.

(4) We conduct extensive experiments on three DARPA E3 datasets. TGL-APT achieves F1-scores of up to 95.7\% while reducing training time, detection latency, and memory usage by approximately 39\%, 33\%, and 22\%, respectively, compared with KAIROS. The investigation results further demonstrate its capability to reconstruct fragmented suspicious activities into interpretable attack stages.

\section{Background and Related Work}
\label{sec:relwork}

\subsection{Research Status}
Existing studies on APT detection and investigation mainly rely on system logs and provenance graphs \cite{TATAM2021e05969,ALAGEEL2026112069}. According to their main objectives, we group existing approaches into three categories: provenance-based APT detection, graph learning and provenance reduction, and attack correlation and investigation.

\subsubsection{Provenance-based APT Detection}
Early provenance-based detection methods mainly rely on expert knowledge, predefined
rules, or statistical deviations from normal system behavior. Heuristic-based approaches
\cite{Holmes2019,Poirot2019,SLEUTH2017,ProvDetector2020,LMTracker2022,Chen2022APT-KGL}
identify suspicious activities using known attack patterns, provenance dependencies, or threat knowledge. Holmes \cite{Holmes2019} maps low-level system events to expert-defined TTPs and constructs high-level attack scenarios, while SLEUTH \cite{SLEUTH2017} reconstructs attack scenarios from audit data. Poirot \cite{Poirot2019} aligns known attack behavior with kernel audit records for threat hunting, and LMTracker \cite{LMTracker2022} models lateral movement paths using heterogeneous graph representations. APT-KGL \cite{Chen2022APT-KGL} further combines threat knowledge with heterogeneous provenance graph learning. Commercial detection systems also commonly organize attack knowledge using MITRE ATT\&CK TTPs \cite{MITRE2022}. APT-MMF \cite{APT-MMF2024} exploits multimodal and multilevel information for APT actor attribution.

Anomaly-based approaches reduce the dependence on predefined attack signatures by
learning normal system behavior \cite{Streamspot2016,Unicorn2020,ThreaTrace2022,R-CAID2024}. StreamSpot \cite{Streamspot2016} performs memory-efficient anomaly detection over streaming heterogeneous graphs, while Unicorn \cite{Unicorn2020} detects abnormal runtime provenance patterns. Related learning-based security classification methods \cite{demirkiran2022ensemble,ccayir2021random} also address detection under highly
imbalanced security data.

Despite their effectiveness, knowledge-driven methods remain dependent on known attack
patterns, whereas conventional anomaly detection has limited ability to capture complex temporal dependencies and multi-hop interactions. These limitations motivate richer graph representations that can model how suspicious activities propagate through provenance dependencies over time.

\subsubsection{Graph Learning and Provenance Reduction}
Graph neural networks have increasingly been used to model complex structural and
temporal dependencies in provenance data. ThreaTrace \cite{ThreaTrace2022} applies
GraphSAGE \cite{GraphSAGE2017} to provenance graphs for node-level threat detection and tracing. MAGIC \cite{Magic2024} employs self-supervised masked graph learning to model benign behavior and accommodate changes in system activity. KAIROS \cite{Kairos2024} uses a temporal graph encoder-decoder to capture provenance evolution and identify anomalous interactions. More recent temporal approaches such as TFLAG \cite{jiang2025tflag} further exploit deviations in historical behavior for APT detection.

Although graph learning improves structural and temporal modeling, long-running
provenance graphs may contain millions of nodes and edges, while attack-relevant evidence remains sparse and embedded in large volumes of routine system activity. Processing such graphs in full therefore incurs substantial computational and memory overhead. Sampling and graph reduction can alleviate this problem
\cite{Zhang2025TAPAS,Li2025TeRed,li2023provgrp}, but reducing provenance solely according to graph size or generic structural properties may remove low-frequency but security-relevant entities or break causal dependencies needed for subsequent analysis.

\subsubsection{Attack Correlation and Investigation}
Beyond detecting suspicious events, APT investigation aims to connect distributed
evidence and reconstruct the underlying attack process. Holmes \cite{Holmes2019} and
Poirot \cite{Poirot2019} correlate suspicious information flows or threat knowledge with system provenance to recover attack-related behavior. APTSHIELD
\cite{zhu2023aptshield} organizes APT behavior into multiple attack stages and generates alerts through rule-based analysis. TAGAPT \cite{cheng2025tagapt} operates on provenance-level attack graphs to support attack-stage analysis. ORTHRUS
\cite{jiang2025orthrus} performs forward and backward provenance analysis to identify
important dependencies for attack attribution. KAIROS \cite{Kairos2024} applies Louvain community detection \cite{Louvain2008} to organize highly anomalous activities into suspicious subgraphs, while Slot \cite{qiao2025slot} employs graph reinforcement learning for provenance-driven APT detection and TTP analysis.

However, long-running APT campaigns may span different processes, files, network
endpoints, and time windows. Related attack activities therefore do not necessarily share the same entities, making entity-based correlation insufficient when attack semantics remain similar but their concrete carriers change. Moreover, suspicious events or subgraphs alone do not reveal the complete causal progression and stage transitions of an attack. Consequently, effective investigation requires both cross-window association of fragmented evidence and causal organization of the resulting suspicious activities.

Overall, existing approaches address provenance detection, graph-scale reduction, and
attack investigation largely as separate problems. Detection models capture increasingly rich system behavior but must process substantial redundant activity; graph reduction improves efficiency but may lose sparse security-relevant evidence; and investigation methods may fail to connect related behaviors when evidence shifts across entities and time windows. These limitations motivate a unified approach that identifies, preserves, tracks, and connects attack-relevant information throughout provenance-based APT analysis.

\subsection{Research Motivation}
The motivation of this study stems from three challenges in provenance-based APT analysis: sparse attack-relevant evidence hidden in large volumes of routine system activity, evolving information relevance that weakens long-term correlation, and fragmented anomalies that provide limited support for attack investigation.

To reduce provenance redundancy without losing critical evidence, graph reduction should consider not only graph size, but also the structural and behavioral relevance of retained entities and the causal context surrounding them. This provides a more suitable basis for efficient temporal graph learning than indiscriminate graph compression.

For long-running analysis, the entities carrying attack-relevant information may change over time, while related activities may appear across different processes, files, network endpoints, and time windows. This requires the analysis focus to adapt to changing node relevance and to associate suspicious activities beyond simple entity co-occurrence.

For attack investigation, isolated anomalies must be further organized according to their causal dependencies and temporal progression. Recovering causal paths and identifying representative stage transitions can transform fragmented suspicious activities into coherent and interpretable attack processes.

Overall, this study focuses on identifying, preserving, tracking, and connecting sparse attack-relevant information throughout provenance-based APT analysis, with the goal of improving detection efficiency, long-range correlation, and attack investigation.

\section{Methodology}
\label{sec:methodology}
TGL-APT detects and investigates stealthy multi-stage APT attacks from large-scale dynamic audit logs. It learns normal system behavior from benign training data without requiring malicious labels and operates on reliably collected audit records, where malicious activities may be concealed within benign behavior. As shown in Figure \ref{fig5.1}, TGL-APT progressively extracts attack-relevant information through graph distillation, adaptive temporal learning, and cross-spatiotemporal alignment, followed by causal reconstruction and attack-stage identification.
\begin{figure*}[htbp]
	\centering
	\includegraphics[width=1\linewidth]{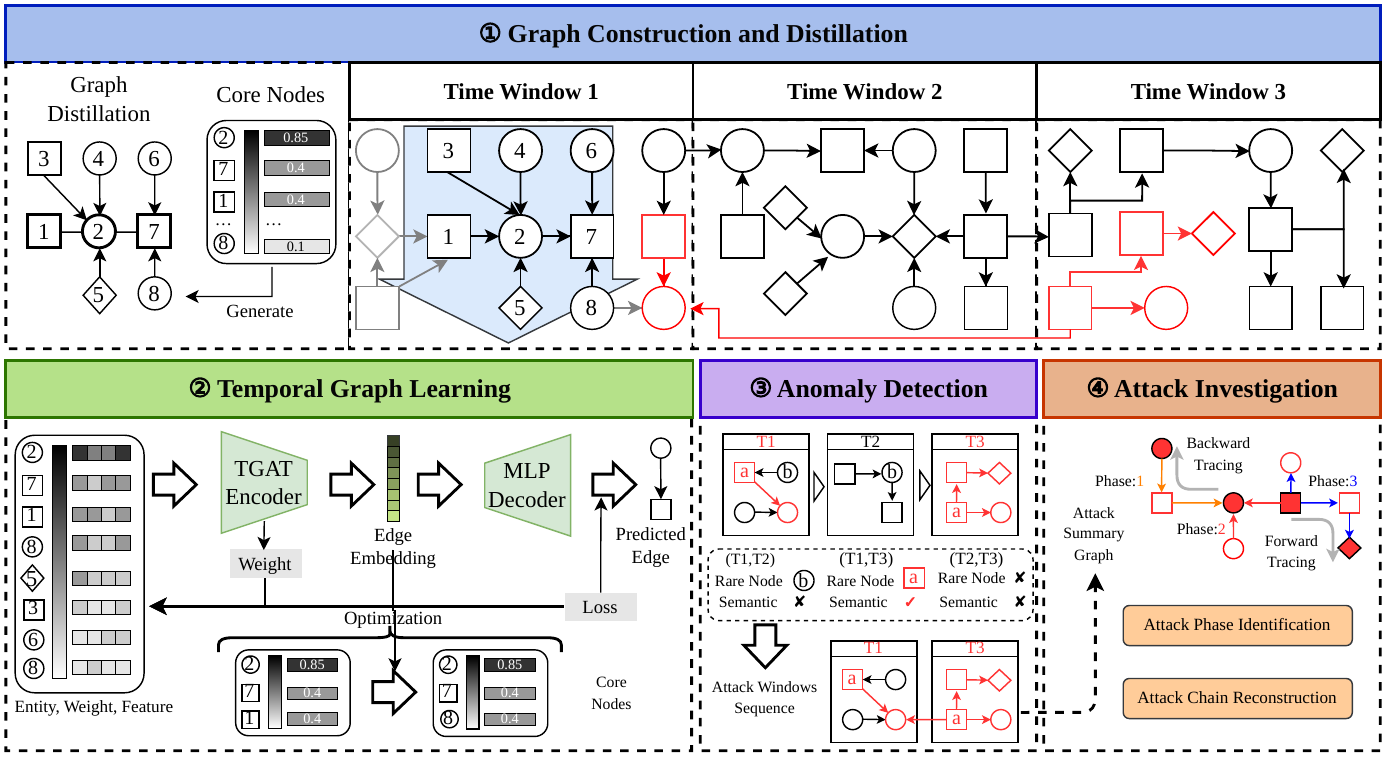}
	\renewcommand{\figurename}{Figure}
	\caption{The architecture of TGL-APT}
	\label{fig5.1}
\end{figure*}

\subsection{Problem Formulation}
\label{pf}
Let $G_{\mathrm{raw}}=(V_{\mathrm{raw}},E_{\mathrm{raw}},T,A)$ denote a dynamic provenance graph containing system entities, dependency edges, timestamps, and attributes, and let $G_{\mathrm{comp}}=\mathcal{D}(G_{\mathrm{raw}})$ denote its distilled representation. The objective of provenance distillation is not merely to reduce graph size, but to remove redundant system information while retaining the evidence required for downstream APT detection and investigation. Let $Y$ represent the security-relevant information needed for downstream APT detection and investigation. We then formulate graph distillation under the Information Bottleneck principle \cite{tishby2000information} as
\begin{equation}
	\mathcal{L}_{\mathrm{IB}}
	=
	I(G_{\mathrm{raw}};G_{\mathrm{comp}})
	-\lambda I(G_{\mathrm{comp}};Y),
\end{equation}
where minimizing $\mathcal{L}_{\mathrm{IB}}$ encourages removing redundant information by minimizing $I(G_{\mathrm{raw}};G_{\mathrm{comp}})$, while preserving task-relevant information by maximizing $I(G_{\mathrm{comp}};Y)$ through the negative second term.

Direct estimation of these mutual-information terms is intractable for large dynamic heterogeneous provenance graphs. We therefore make this objective operational using tractable proxies for redundancy and task relevance. At the node level, task relevance is characterized from two complementary aspects: structural mediation $C(v)$, which captures the role of node $v$ in connecting provenance dependencies, and behavioral specificity $R(v)$, which captures the distinctiveness of its interaction pattern. We define
\begin{equation}
	Score(v)=\alpha C(v)+\beta R(v),
\end{equation}
which provides a tractable measure of node-level information relevance.

Accordingly, the practical distillation objective is formulated as
\begin{equation}
	\min_{G_{\mathrm{comp}}\subseteq G_{\mathrm{raw}}}
	\Omega(G_{\mathrm{comp}})
	-\mu
	\sum_{v\in V_{\mathrm{comp}}} Score(v),
	\label{eq:distillation_objective}
\end{equation}
subject to
\begin{equation}
	\left|
	H_{\mathrm{struct}}(G_{\mathrm{raw}})
	-
	H_{\mathrm{struct}}(G_{\mathrm{comp}})
	\right|
	\leq \epsilon_H,
\end{equation}
and the preservation of temporally valid causal reachability among retained high-priority nodes. Here, $\Omega(G_{\mathrm{comp}})$ measures the size of the retained graph, and $H_{\mathrm{struct}}$ denotes the structural entropy defined in the following subsection. 

\textbf{Remark 1 (Connection to the IB Objective).}
Under the assumption that attack-relevant information is sparse and propagates along temporally valid provenance dependencies, $C(v)$ and $R(v)$ provide complementary surrogates for a node's contribution to $I(G_{\mathrm{comp}};Y)$. A larger $C(v)$ indicates stronger causal mediation, whereas a larger $R(v)$ indicates less redundant and more behaviorally distinctive information. More specifically, we assume a monotonic relation in which, for nodes with comparable local context, stronger causal mediation or behavioral specificity does not decrease their marginal contribution to task-relevant information. Under this assumption, the score term in Eq.~\eqref{eq:distillation_objective} serves as a monotonic ranking surrogate for preserving $I(G_{\mathrm{comp}};Y)$, rather than as a numerical lower bound or exact estimator of mutual information.

Thus, graph-size reduction operationalizes the compression side of the IB objective, while the node-priority score operationalizes its task-relevance side; the structural and causal constraints further preserve the provenance context required for downstream analysis.

\subsection{Provenance Graph Distillation}
\label{subsec:compressed_view}

Based on the information-bottleneck formulation in Section~\ref{pf}, we first formalize the semantic role of the nodes that guide provenance graph distillation.

\textbf{Definition 1 (Information-Bottleneck Node).}
Given the provenance history $G_{\leq t}$ and downstream security-relevant information $Y$, a node $v$ is an information-bottleneck (IB) node at time $t$ if its local provenance context $Z_t(v)$, consisting of its temporally valid dependencies and interaction history, satisfies three semantic properties. First, from an information-theoretic perspective, $Z_t(v)$ acts as an approximately minimal sufficient statistic for $Y$, preserving security-relevant information while suppressing redundant provenance context, i.e.,
\[
I(Y;G_{\leq t}\mid Z_t(v)) \leq \epsilon,
\]
while $I(Z_t(v);G_{\leq t})$ is kept small. Second, from a graph-topological perspective, $v$ acts as a causal bridge whose removal may disrupt temporally valid causal reachability. Third, from a temporal perspective, the relevance of $v$ may change as system behavior evolves, and its behavioral change may indicate a transition toward an attack-relevant state. Therefore, an IB node is characterized by information sufficiency, causal mediation, and temporal transition relevance rather than by structural importance alone.

We construct $G_{\mathrm{raw}}$ from audit records collected from Windows ETW, Linux Auditd, and other sources. The graph contains three types of system entities: processes, files, and network sockets. Directed edges represent system dependencies, while $T$ and $A$ provide temporal and attribute information, respectively. The detailed entities, attributes, and dependency relationships are summarized in Table~\ref{tab5.1}.
\begin{table}[htbp]
	\centering
	\caption{System Entities, Their Attributes, and Dependency Relationships}
	\label{tab5.1}
	\footnotesize
	\setlength{\tabcolsep}{2pt}
	\begin{tabular}{ccc p{2cm}}
		\toprule
		Source Node & Target Node & Relation & Entity Attribute \\
		\midrule
		\multirow{3}{*}{Process}
		& Process & Start, Close, Clone        & Image pathname \\
		& File    & Read, Write, Open, Exec    & File pathname \\
		& Socket  & Send, Receive              & Src/Dst IP/port \\
		\bottomrule
	\end{tabular}
\end{table}

Since directly estimating the mutual-information conditions in Definition~1 is intractable for large dynamic provenance graphs, we operationalize IB-node identification using structural mediation and behavioral specificity as tractable surrogates. The distillation priority of node $v_i$ is defined as
\begin{equation}
	S_{\text{core}}(v_i) = \alpha \cdot C(v_i) + \beta \cdot R(v_i), \quad \alpha + \beta = 1, \beta > \alpha,
\end{equation}
where $C(v_i)$ measures structural importance by combining Betweenness Centrality and PageRank, and $R(v_i)$ characterizes behavioral specificity using Inverse Interaction Frequency and interaction entropy. A larger weight is assigned to $R(v_i)$ to reduce the bias of global centrality toward frequently active benign entities and to emphasize low-frequency, behaviorally distinctive interactions.

The behavioral specificity term helps retain low-frequency and locally connected entities that may be overlooked by global centrality. Low-ranked intermediate nodes required for causal connectivity are further preserved during the subsequent pruning stage.

To control structural distortion introduced by graph reduction, we use a spectral structural entropy derived from the normalized graph Laplacian. For a graph $G$, let $\lambda_1,\ldots,\lambda_n$ denote the eigenvalues of its normalized Laplacian matrix. We first normalize the spectrum as
\begin{equation}
	p_i = \frac{\lambda_i}{\sum_{j=1}^{n}\lambda_j},
\end{equation}
and define the structural entropy as
\begin{equation}
	H_{\text{struct}}(G)
	=
	-\sum_{i=1}^{n} p_i \log p_i,
\end{equation}
where $p_i \log p_i$ is defined as zero when $p_i=0$.

Based on the node priorities, TGL-APT performs a two-stage distillation process. In the knowledge screening stage, nodes are ranked according to $S_{\text{core}}(v_i)$, and the top-ranked nodes are retained as the candidate set. In the structure pruning stage, redundant nodes are removed subject to two constraints. First, the causal connectivity constraint preserves at least one temporally valid causal path between retained nodes that are causally connected in $G_{\text{raw}}$; low-ranked intermediate nodes are retained when required to maintain such paths. Second, the structural information constraint requires the structural entropy deviation between the raw and distilled graphs to remain below a predefined threshold $\epsilon_H$. Together, these constraints preserve the structural and causal context surrounding the selected information-bottleneck nodes.

The node-priority score determines which entities are considered during
knowledge screening, whereas the subsequent constraints determine whether
their surrounding provenance context can be safely reduced. The resulting
preservation and downstream stability properties are summarized in
Lemma~\ref{lemma:distillation_preservation}.

\begin{lemma}[Structural-Causal Preservation and Detection Stability]
	\label{lemma:distillation_preservation}
	Let $G_{\mathrm{comp}}$ be the graph produced by the two-stage
	distillation process. Let $\mathcal{P}_{\mathrm{req}}$ denote the causal
	relations between retained nodes connected by at least one temporally valid
	path in $G_{\mathrm{raw}}$, and $\mathcal{P}_{\mathrm{pres}}$ denote the
	subset whose causal reachability is preserved in $G_{\mathrm{comp}}$.
	Define
	\begin{equation}
		\rho_{\mathrm{causal}}
		=
		\frac{|\mathcal{P}_{\mathrm{pres}}|}
		{|\mathcal{P}_{\mathrm{req}}|},
	\end{equation}
	where $\rho_{\mathrm{causal}}=1$ when
	$\mathcal{P}_{\mathrm{req}}=\emptyset$, and
	\begin{equation}
		\mathcal{L}_{\mathrm{info}}
		=
		\left|
		H_{\mathrm{struct}}(G_{\mathrm{raw}})
		-
		H_{\mathrm{struct}}(G_{\mathrm{comp}})
		\right|
		+
		\kappa(1-\rho_{\mathrm{causal}}).
	\end{equation}
	
	If the structural entropy deviation is bounded by $\epsilon_H$ and
	causal reachability is preserved for every relation in
	$\mathcal{P}_{\mathrm{req}}$, then
	\begin{equation}
		\mathcal{L}_{\mathrm{info}}
		(G_{\mathrm{raw}},G_{\mathrm{comp}})
		\leq \epsilon_H.
	\end{equation}
	
	Furthermore, let $Z_e(G)$ denote the temporal representation of a retained
	event $e$. If the temporal encoder is $L_{\Phi}$-Lipschitz with respect to
	$\mathcal{L}_{\mathrm{info}}$ and the reconstruction-error function is
	$L_{\ell}$-Lipschitz with respect to $Z_e$, then
	\begin{equation}
		\left|
		\mathrm{RE}(e;G_{\mathrm{raw}})
		-
		\mathrm{RE}(e;G_{\mathrm{comp}})
		\right|
		\leq
		L_{\ell}L_{\Phi}\epsilon_H.
	\end{equation}
\end{lemma}

\begin{proof}
	The causal-connectivity constraint implies
	$\mathcal{P}_{\mathrm{pres}}=\mathcal{P}_{\mathrm{req}}$ and hence
	$\rho_{\mathrm{causal}}=1$. Together with the structural-entropy
	constraint,
	\[
	\mathcal{L}_{\mathrm{info}}
	=
	\left|
	H_{\mathrm{struct}}(G_{\mathrm{raw}})
	-
	H_{\mathrm{struct}}(G_{\mathrm{comp}})
	\right|
	\leq \epsilon_H.
	\]
	
	For a retained event $e$, the Lipschitz assumptions further give
	\[
	\begin{aligned}
		&
		\left|
		\mathrm{RE}(e;G_{\mathrm{raw}})
		-
		\mathrm{RE}(e;G_{\mathrm{comp}})
		\right| \\
		&\leq
		L_{\ell}
		\left\|
		Z_e(G_{\mathrm{raw}})
		-
		Z_e(G_{\mathrm{comp}})
		\right\|
		\leq
		L_{\ell}L_{\Phi}
		\mathcal{L}_{\mathrm{info}}
		\leq
		L_{\ell}L_{\Phi}\epsilon_H.
	\end{aligned}
	\]
	
	Thus, bounded structural-causal distortion induces a bounded perturbation
	in the reconstruction-error scores used for downstream anomaly detection.
\end{proof}

Lemma~\ref{lemma:distillation_preservation} shows that the proposed
distillation not only bounds structural and causal information loss, but
also limits its influence on downstream anomaly scores under the stated
stability assumptions. The resulting distilled graph
$G_{\mathrm{comp}}=(V_{\mathrm{comp}},E_{\mathrm{comp}})$ is then used
for subsequent temporal graph learning.

\subsection{Temporal Graph Learning with Adaptive Focus Mechanism}
The temporal graph learning module aims to construct a representation learning framework capable of simultaneously capturing the structural evolution and temporal dependencies of the provenance graph. This module adopts an encoder-decoder architecture, learning normal system behavior patterns using only benign training data while dynamically refining its analysis focus as observed system behavior evolves. As shown in Figure \ref{fig5.1}, when a new edge $e_t$ appears in the compressed view $G_{\text{comp}}$, the encoder generates an edge embedding $\mathbf{Z}$ based on the neighborhood state $\mathbf{s}_{t^-}$ before the edge's occurrence, and the decoder reconstructs the edge type probability distribution $\mathbf{P}(e_t)$ from this embedding. The encoder and decoder are jointly trained exclusively on benign data, aiming to minimize the reconstruction error for benign edges.

The core of the encoder is to construct an edge embedding that integrates spatiotemporal information. TGL-APT employs an enhanced Temporal Graph Network (TGN) architecture, incorporating a multi-head graph attention mechanism to focus on potentially anomalous interaction patterns while capturing the evolution of node historical states.

For a new edge $e_t = (v_{src}, v_{dst})$ occurring at time $t$, the computation of its embedding $\mathbf{Z}_t$ depends not only on the historical states $\mathbf{s}_{t^-}(v_{src})$ and $\mathbf{s}_{t^-}(v_{dst})$ of the source and target nodes but also on the local neighborhood structure obtained through sampling. Unlike the base TGN, TGL-APT utilizes the set of core nodes $H$ selected in Section \ref{subsec:compressed_view} to guide neighborhood sampling, prioritizing connections to core nodes, thus more efficiently capturing critical paths that attacks might exploit. The encoding process is as follows:
\begin{equation}
	\mathbf{Z}_t = \text{GNN}_{\theta}(\mathbf{s}_{t^-}(v_{src}), \mathbf{s}_{t^-}(v_{dst}), \mathcal{N}_{H}(v_{src}), \mathcal{N}_{H}(v_{dst}), \mathbf{e}, \mathbf{t}),
\end{equation}
where $\mathcal{N}_{H}(\cdot)$ denotes biased neighborhood sampling based on the core node set $H$, $\mathbf{e}$ is the feature vector of edges within the neighborhood (formed by concatenating node feature hashes and edge types), and $\mathbf{t}$ is the timestamp vector of the corresponding edges. The GNN model employs a multi-head attention mechanism for message passing, enabling the model to assign differentiated importance weights to different neighbor interactions, thereby highlighting anomalous or rare connection patterns.

After a new edge $e_t$ occurs, the states of the source node $v_{src}$ and target node $v_{dst}$ need to be updated accordingly. TGL-APT uses a Gated Recurrent Unit (GRU) to iteratively update the states of both endpoints, enabling them to retain memory of their historical interaction sequences:
\begin{equation}
	\begin{aligned}
		\mathbf{s}_t(v_{src}) &= \mathrm{GRU}_{\phi}
		\left(\mathbf{s}_{t^-}(v_{src}), \mathbf{Z}_t\right), \\
		\mathbf{s}_t(v_{dst}) &= \mathrm{GRU}_{\phi}
		\left(\mathbf{s}_{t^-}(v_{dst}), \mathbf{Z}_t\right).
	\end{aligned}
\end{equation}

The updated state $\mathbf{s}_t(\cdot)$ is then used for subsequent edge-embedding computations. 

The decoder is an MLP whose task is to reconstruct, based on the edge embedding $\mathbf{Z}_t$ output by the encoder, the probability distribution $\mathbf{P}(e_t)$ over the set of predefined interaction types for edge $e_t$:
\begin{equation}
	\mathbf{P}(e_t) = \text{MLP}_{\psi}(\mathbf{Z}_t).
\end{equation}

During training, TGL-APT optimizes the model parameters $\theta, \phi, \psi$ by minimizing the cross-entropy loss between the predicted distribution $\mathbf{P}(e_t)$ and the one-hot encoding $\mathbf{L}(e_t)$ of the true edge type:
\begin{equation}
	\mathcal{L} = \text{CrossEntropy}(\mathbf{P}(e_t), \mathbf{L}(e_t)).
\end{equation}

The model thus learns what type of interaction is considered "normal" within a given spatiotemporal context. During inference, the reconstruction error is this cross-entropy loss value, reflecting the deviation of the current event from the learned normal patterns.

In long-running provenance streams, the relevance of information-bottleneck nodes may change as system behavior evolves. A core node set selected from the initial graph therefore cannot remain equally informative over time. To address this issue, TGL-APT continuously refines the core set using model-derived signals from temporal learning. Attention contribution reflects how strongly a node influences current neighborhood representations, while embedding deviation measures how far its current behavior departs from the benign prototype. Together, these signals allow the model to incorporate newly informative nodes and discard nodes whose contribution has diminished.
\begin{algorithm}
	\small
	\caption{Adaptive Focus Mechanism}\label{alg5.1}
	\begin{algorithmic}[1]
		\Require Current core node set $H_{\text{curr}}$, graph attention weights, node embedding vectors, benign embedding cluster centers, total number of nodes $N$
		\Ensure Updated core node set $H_{\text{new}}$
		
		\State Compute the average graph attention weight for all nodes over the past epoch
		\State Sort in descending order by weight, take the top $\eta_{\alpha}$ nodes and add to $H_{\text{add}}^{\alpha}$
		
		\State Compute the Euclidean distance between all node embeddings and the benign prototype embedding cluster center
		\State Sort in descending order by deviation, take the top $\eta_{\delta}$ nodes and add to $H_{\text{add}}^{\delta}$
		
		\For{$v \in H_{\text{curr}}$}
		\State Compute the average attention contribution $C_{\text{att}}(v)$ for node $v$
		\EndFor
		\State Sort by $C_{\text{att}}(v)$ in ascending order, remove the bottom $\eta_{\text{drop}}$ nodes, the remaining nodes form $H_{\text{keep}}$
		
		\State $H_{\text{new}} \leftarrow (H_{\text{add}}^{\alpha} \cup H_{\text{add}}^{\delta}) \cup H_{\text{keep}}$
		\State $M_{\text{max}} \leftarrow \gamma \cdot N$ 
		\If{$|H_{\text{new}}| > M_{\text{max}}$}
		\State Compute the comprehensive score $S(v)$ for each node in $H_{\text{new}}$
		\State Sort by $S(v)$ in descending order, retain the top $M_{\text{max}}$ nodes
		\EndIf
		
		\State \Return $H_{\text{new}}$
	\end{algorithmic}
\end{algorithm}

The adaptive update follows a relevance-refinement process. Nodes with high attention contribution or large embedding deviation are treated as newly informative candidates, while low-contribution nodes are removed from the current core set. The updated set is further constrained by a proportional size budget to prevent uncontrolled growth. In this way, adaptive focusing extends the initial graph distillation from a static selection to a dynamic information focus that follows changes in node relevance over time. 

\subsection{Anomaly Detection}
The anomaly detection module aims to identify coherent suspicious behavior sequences from continuous system activities. Its core design is \textbf{cross-spatiotemporal attack fingerprint alignment}, which addresses the \textbf{synonymous but heterogeneous} problem, where behaviorally related attack activities may appear across different entities and time windows without direct entity co-occurrence. TGL-APT therefore complements entity-based association with semantic similarity between suspicious time windows to recover such fragmented associations. This module operates at the time-window level, while attack stages are subsequently identified during attack investigation.

TGL-APT divides the temporal graph into fixed-length time windows $T$. For each window, TGL-APT identifies a set of suspicious nodes $\mathcal{S}_T$. A node $v$ is deemed suspicious if it simultaneously satisfies the following two conditions:

(1) Anomaly: Node $v$ must be the endpoint of at least one anomalous edge. Whether an edge $e$ is anomalous depends on the reconstruction error $\text{RE}(e)$ produced by the temporal graph learning module. TGL-APT dynamically determines a reconstruction error threshold $\sigma_T$ from the reconstruction-error distribution of each time window. If $\text{RE}(e) > \sigma_T$, edge $e$ is considered anomalous.

(2) Rarity: The system entity represented by node $v$ should be uncommon in normal system behavior. TGL-APT uses Inverse Document Frequency (IDF) to quantify the global rarity of a node. For a node $v$ appearing in $N_v$ historical time windows, its IDF value is: $\text{IDF}(v) = \ln\left( \frac{N}{N_v + 1} \right)$, where $N$ is the total number of windows. A higher IDF value indicates rarer nodes. A rarity threshold $\alpha$ is set; node $v$ satisfies the rarity condition only if $\text{IDF}(v) > \alpha$.

A node meeting both the anomaly and rarity conditions is included in the current window's suspicious node set $\mathcal{S}_T$. 

Entity co-occurrence provides a direct way to associate suspicious windows, but it may miss related activities when the involved entities change over time. To complement this signal, TGL-APT constructs an attack fingerprint for each suspicious time window and measures semantic similarity between window fingerprints. Each time window $T$ is regarded as a semantic document, and the suspicious nodes $S_T$ are treated as its behavioral words. Unlike general semantic association over all system entities, the fingerprint is constructed only from suspicious activities and weighted according to their discriminative importance, focusing the association on attack-relevant behavioral evidence.

Different from the simple bag-of-words model, TGL-APT uses a TF-IDF weighted mechanism to generate attack fingerprint vectors, where the term frequency $\text{TF}(v, T)$ is determined by the anomaly intensity of edges connected to node $v$, so as to highlight the contribution of key suspicious entities to the attack fingerprint.

To reduce high-dimensional sparsity, contrastive learning is used to map TF-IDF weighted attack fingerprints into a low-dimensional embedding space. During training, windows sharing suspicious entities are treated as positive pairs, while temporally distant windows without shared entities are treated as negative pairs. Entity co-occurrence therefore serves as a weak supervisory signal for learning semantic representations of suspicious windows. The learned embeddings capture behavioral similarity between windows, allowing related suspicious activities to be associated even when they do not share the same entities.

During inference, the semantic similarity between window embeddings is measured by cosine similarity. The final window association adopts a fusion rule that combines entity co-occurrence and semantic similarity: $T_i$ and $T_j$ are associated if $\mathcal{S}_i \cap \mathcal{S}_j \neq \emptyset$ or $\text{Sim}_{\text{sem}}(T_i, T_j) > \kappa$. This dual-channel association complements exact entity matching with fingerprint-level semantic similarity, allowing fragmented suspicious activities across different entities and time windows to be linked for subsequent attack investigation.

Each queue $Q$ maintains a dynamically updated anomaly score $\text{AnomalyScore}(Q)$. This score is the product of the anomaly scores of all time windows in the queue, where the anomaly score for a single window $T$ is defined as the average reconstruction error of all anomalous edges within that window.
\begin{equation}
	\text{AnomalyScore}(Q) = \prod_{i=1}^{n} \text{AnomalyScore}(T_i).
\end{equation}

TGL-APT presets a queue anomaly threshold $\beta$ based on benign validation data. At runtime, when the anomaly score of a queue $Q$ exceeds $\beta$, that queue is flagged as a malicious activity sequence, triggering an investigation alert. This sequence $\{T_1, T_2, ..., T_n\}$ represents the detected unit of potentially malicious activity that is coherent in time and semantics.

\subsection{Attack Investigation}
The core task of attack chain reconstruction is to restore a coherent and concise attack path from the detected malicious activity sequence $Q$. TGL-APT employs a lightweight reconstruction strategy based on causal expansion of anomalous nodes, directly utilizing the suspicious nodes identified during the anomaly detection stage and their causal connections in the original provenance graph to rebuild the attack flow, avoiding the computational overhead and logical redundancy associated with complex community detection and subgraph fusion.

First, from all time windows contained in the malicious activity sequence $Q$, all nodes marked as "suspicious" are extracted to form the initial anomalous node set $V_{\text{anomaly}}$. For each node $v$ in $V_{\text{anomaly}}$, a bidirectional causal exploration is performed in the original full provenance graph $G_{\text{raw}}$.

(1) Forward Causal Expansion: Search forward in time along causal edges (e.g., process creation, file read/write, network connection) starting from $v$, within a reasonable time range $\Delta T$, and add all reachable nodes to the set $V_{\text{forward}}(v)$.

(2) Backward Traceability Reconstruction: Search backward in time along dependency paths (e.g., parent process, file dependency) that can reach $v$ within $\Delta T$, and add all such nodes to the set $V_{\text{backward}}(v)$. The expanded nodes and edges are first merged into a preliminary candidate subgraph. To suppress irrelevant expansion, causal branches that neither connect suspicious seed nodes nor lie on a temporally valid causal path between them are removed. The remaining nodes and edges form $G_{\text{candidate}}$. A timestamp-constrained path search is then performed on $G_{\text{candidate}}$ to extract the causal backbone $C$, while redundant side branches are excluded from the reconstructed attack chain.
\begin{figure}[htbp]
	\centering
	\includegraphics[width=1\linewidth]{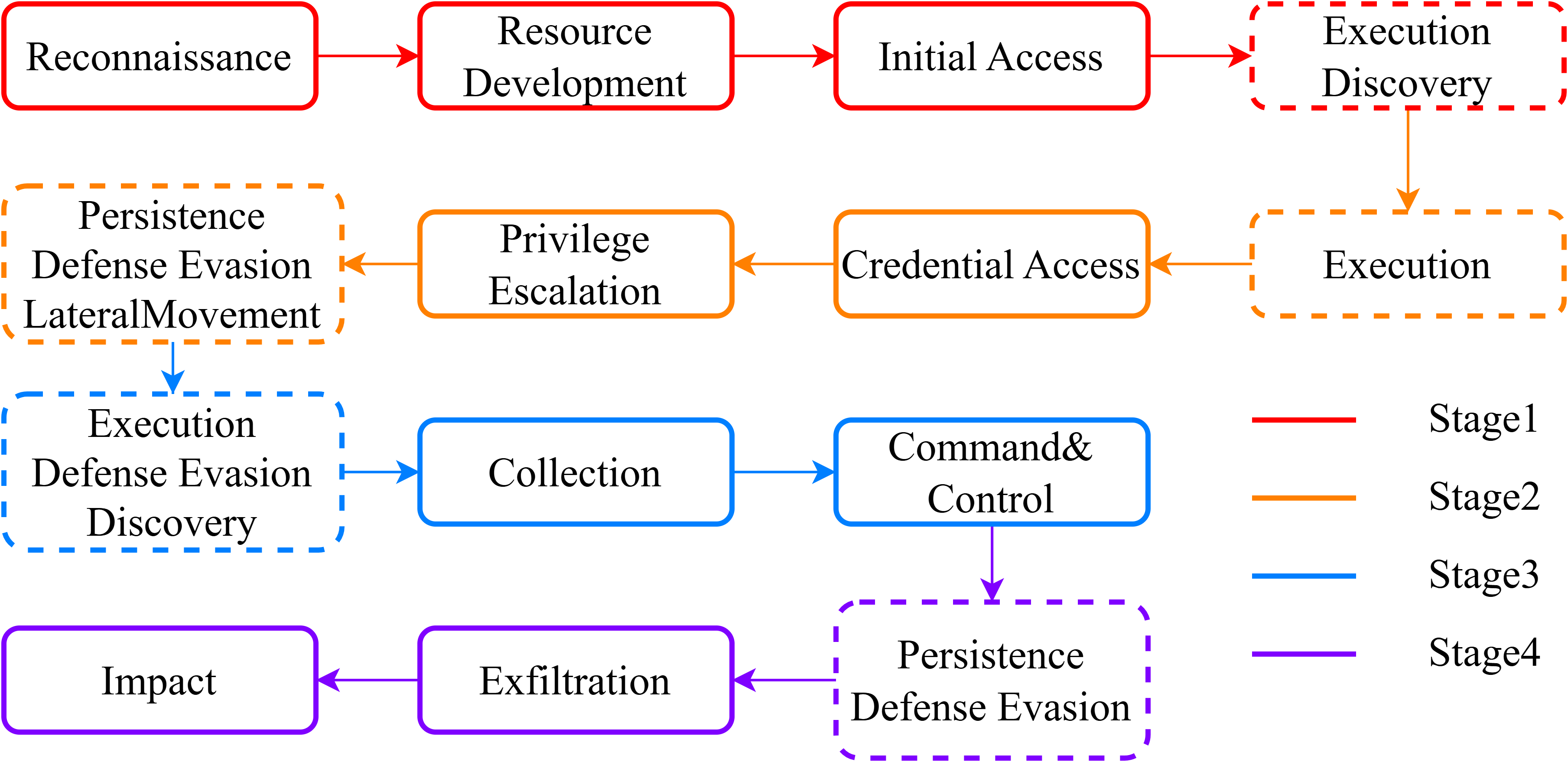}
	\renewcommand{\figurename}{Figure}
	\caption{Four-stage division of APT attacks and their tactical dependencies}
	\label{fig5.2}
\end{figure}
As shown in Figure \ref{fig5.2}, following TAGAPT~\cite{cheng2025tagapt}, we organize the MITRE ATT\&CK tactics into four attack stages rather than directly adopting the Cyber Kill Chain. The stage boundaries are automatically determined by the timestamps of the identified Stage Characterization Nodes (SCNs).

TGL-APT proposes an automated phase categorization method based on structural pivotalness and behavioral specificity, centered around identifying Stage Characterization Nodes (SCNs). SCNs are typically process nodes that play a dominant role within a stage, characterized by three features. Feature 1, Global Pivotalness: Act as key connection points across multiple functional flows within that stage. Feature 2, Local Connectivity: Possess high out-degree and in-degree, actively initiating or passively receiving various types of operations. Feature 3, Behavioral Specificity: Their behavior patterns (types of functional flows they participate in) are significantly distinct from other SCNs, representing different attack tactics. Based on this, TGL-APT designs the following automated phase categorization algorithm:

TGL-APT derives temporally valid functional flows from the pruned candidate graph $G_{\text{candidate}}$ to form a set $F$. Each functional flow represents a causally connected sequence of activities corresponding to a candidate attack sub-goal. It then calculates the number of functional flows $|F_{a_i}|$ carried by each process node $a_i$ in $\mathcal{C}$ and its local graph density, screening nodes that simultaneously satisfy high pivotalness and high connectivity to form a candidate set $P_{\text{can}}$. To avoid misclassifying processes performing similar tasks as different stage starts, the functional flow similarity between nodes $a_i$ and $a_j$ is computed:
\begin{equation}
	Sim(a_i, a_j) = \frac{|F_{a_i} \cap F_{a_j}|}{\max(|F_{a_i}|, |F_{a_j}|)}.
\end{equation}

If $Sim(a_i,a_j)\geq\theta_{\mathrm{sim}}$, the nodes are considered to carry similar functional flows, and the more pivotal node is retained, merging redundant candidates. Finally, nodes are sorted by their comprehensive pivotalness score, and the top $K_{\max}$ nodes are selected as stage characterization nodes $P_{\text{scn}}$. The timestamps of these nodes in the attack chain are used as boundaries to partition the different attack stages $\text{Phases}$. This method effectively identifies key nodes that signify stage transitions, enabling automated phase categorization of complex APT attacks and assisting security analysts in better understanding attack tactic evolution and formulating response strategies.

\section{Evaluation}
\label{sec:eval}
To comprehensively evaluate the effectiveness, efficiency, and practicality of TGL-APT for APT detection and investigation, we conduct systematic experiments from multiple perspectives, including detection performance, comparison with existing methods, computational efficiency, hyperparameter sensitivity, and attack reconstruction capability. Specifically, this section aims to answer the following key research questions through experiments:
(1) Can TGL-APT accurately detect attack behaviors in running systems? (2) How does TGL-APT compare with state-of-the-art techniques? (3) How efficient is TGL-APT in terms of training time, detection latency, and memory consumption? (4) How do hyperparameters affect the detection and runtime performance of TGL-APT? (5) Can TGL-APT accurately reconstruct attack behaviors from the original provenance graph?

A prototype of TGL-APT was implemented in Python. scikit-learn was used for hierarchical feature hashing and PyG for TGL-APT's graph learning framework. Finally, GraphViz was used to visualize the attack summary graphs for manual inspection. Experiments were conducted on a Linux workstation with 64GB RAM, an Intel Core Ultra 7 processor (3.9 GHz), and an NVIDIA RTX A5000 48GB GPU.

\subsection{Experimental Setup}
\paragraph{Datasets and Data Splits.}
This section evaluates TGL-APT on three public datasets from DARPA Engagement 3: Cadets, ClearScope, and Theia \cite{DARPAe3}. These datasets differ in granularity, log format, attack methods, and system environments, providing diverse settings for evaluating TGL-APT and comparing it with state-of-the-art APT detection methods.

The DARPA E3 dataset was collected from enterprise networks during red-blue adversarial exercises in the DARPA Transparent Computing project \cite{DARPAe3}. It provides attack scenarios, tools, attack details, example attack graphs, and system logs. The attacks involve activities such as HTTP communication, privilege elevation, backdoor execution, and .dll or .so injection across different hosts and endpoints. We use the Cadets, ClearScope, and Theia sub-datasets and extract the nodes and edges required for model training and anomaly detection. The datasets are divided chronologically into training, validation, and testing periods to preserve temporal order. Events are further divided into consecutive non-overlapping windows with a length and step size of 15 minutes. The detailed dataset splits and temporal window settings are summarized in Table~\ref{tab:dataset_settings}.
\begin{table}[htbp]
	\centering
	\caption{Dataset splits and temporal window settings.}
	\label{tab:dataset_settings}
	\footnotesize
	\setlength{\tabcolsep}{3pt}
	\begin{tabular}{cccccc}
		\toprule
		Dataset & Training & Validation & Testing & Window & Step \\
		\midrule
		Cadets    & Apr. 2--4 & Apr. 5 & Apr. 6--7   & 15 min & 15 min \\
		ClearScope & Apr. 4--6 & Apr. 7 & Apr. 10--11 & 15 min & 15 min \\
		Theia      & Apr. 3--5 & Apr. 9 & Apr. 10--12 & 15 min & 15 min \\
		\bottomrule
	\end{tabular}
\end{table}

\paragraph{Implementation Details.}
TGL-APT is trained using the Adam optimizer with a learning rate of $5\times10^{-5}$ and a temporal graph embedding dimension of 100. The model is trained for 50 epochs on Cadets and Theia and for 30 epochs on ClearScope. All trainable parameters are initialized using the default initialization schemes of their corresponding PyTorch
layers. For anomaly detection, the reconstruction-error threshold is dynamically calculated for each temporal window as $\sigma_T=\mu_T+0.6s_T$, where $\mu_T$ and $s_T$ denote the mean and standard deviation, respectively, of the reconstruction errors of all edges within the current window. An edge is identified as anomalous
when its reconstruction error exceeds $\sigma_T$. The functional-flow similarity threshold $\theta_{\mathrm{sim}}$ is empirically set to 0.54.

\subsection{Detection Performance}
To evaluate the fundamental anomaly detection capability of the TGL-APT framework, detection performance experiments were first conducted on three public datasets. This experiment aims to answer research question (1), i.e., whether TGL-APT can accurately identify attack behaviors from system operation logs. The widely recognized confusion matrix metrics were employed for evaluation, including True Positives (TP), True Negatives (TN), False Positives (FP), and False Negatives (FN). Based on these, Precision, Recall, Accuracy, and F1-Score were calculated to comprehensively measure the model's overall performance under imbalanced positive and negative samples.

The experimental results are shown in Table \ref{tab5.2}. On the three real-world DARPA E3 datasets, TGL-APT maintained a high recall rate, meaning nearly all attack events were successfully captured. The precision was slightly lower, at 0.9167, 0.8333, and 0.8889 on the three datasets respectively, due to a very small number of false positives. Nevertheless, TGL-APT maintained strong F1-scores across all datasets, demonstrating a favorable balance between precision and recall while keeping the false positive rate low.
\begin{table}[htbp]
	\centering
	\renewcommand{\tablename}{Table}
	\caption{Detection Performance of TGL-APT}
	\label{tab5.2}
	\setlength{\tabcolsep}{3pt}
	\begin{tabular}{c|cccc|cccc}
		\toprule
		Dataset & TP & TN & FP & FN & Precision & Recall & Accuracy & F1 \\ 
		\midrule
		Cadets            & 11          & 328         & 1           & 0           & 0.9167             & 1.0000          & 0.9971            & 0.9565      \\
		ClearScope        & 5           & 113         & 1           & 0           & 0.8333             & 1.0000          & 0.9916            & 0.9091      \\
		Theia             & 8           & 217         & 1           & 1           & 0.8889            & 0.8889          & 0.9912            & 0.8889      \\ \bottomrule
	\end{tabular}
\end{table}

\subsection{Comparison Experiment}
To further evaluate the detection performance of TGL-APT, we compare it with three representative provenance-based APT detection methods, KAIROS \cite{Kairos2024}, TFLAG \cite{jiang2025tflag}, and ORTHRUS \cite{jiang2025orthrus}. This experiment aims to answer research question (2), i.e., how TGL-APT performs compared with existing APT detection approaches. The experimental results are shown in Table \ref{tab5.3}. Overall, TGL-APT achieves strong detection performance across the three datasets, particularly in recall and F1-score.
\begin{table}[htbp]
	\centering
	\renewcommand{\tablename}{Table}
	\caption{Comparison Results of TGL-APT with Other Methods}
	\label{tab5.3}
	\setlength{\tabcolsep}{3pt}
	\begin{tabular}{ccccccc}
		\toprule
		Dataset & Method & Precision & Recall & Accuracy & F1 & FPR \\
		\midrule
		
		\multirow{4}{*}{Cadets}
		& KAIROS  & 0.8750 & 0.6364 & 0.9853 & 0.7368 & 0.0030 \\
		& TFLAG   & 0.8889 & 0.7273 & 0.9882 & 0.8000 & 0.0030 \\
		& ORTHRUS & 0.0005 & 0.0769 & 0.9897 & 0.0010 & \textbf{0.0101} \\
		& TGL-APT & \textbf{0.9167} & \textbf{1.0000} & \textbf{0.9971} & \textbf{0.9565} & 0.0030 \\
		\midrule
		
		\multirow{4}{*}{ClearScope}
		& KAIROS  & 0.7143 & \textbf{1.0000} & 0.9832 & 0.8333 & 0.0175 \\
		& TFLAG   & \textbf{1.0000} & 0.8000 & \textbf{0.9946} & 0.8889 & \textbf{0.0000} \\
		& ORTHRUS & 0.7500 & 0.6000 & 0.9748 & 0.6667 & 0.0088 \\
		& TGL-APT & 0.8333 & \textbf{1.0000} & 0.9916 & \textbf{0.9091} & 0.0088 \\
		\midrule
		
		\multirow{4}{*}{Theia}
		& KAIROS  & 0.8182 & \textbf{1.0000} & \textbf{0.9912} & \textbf{0.9000} & 0.0092 \\
		& TFLAG   & 0.7143 & 0.7143 & 0.9680 & 0.7143 & 0.0169 \\
		& ORTHRUS & 0.0000 & 0.0000 & 0.9604 & 0.0000 & \textbf{0.0000} \\
		& TGL-APT & \textbf{0.8889} & 0.8889 & \textbf{0.9912} & 0.8889 & 0.0046 \\
		\bottomrule
	\end{tabular}
\end{table}

On Cadets, TGL-APT achieved the highest recall and F1-score while maintaining high precision and a low false positive rate. Compared with KAIROS and TFLAG, its F1-score increased by approximately 21.9 and 15.7 percentage points, respectively. On ClearScope, TGL-APT achieved a recall of 1.0000 and the highest F1-score of 0.9091, providing a better balance between precision and recall. On Theia, TGL-APT achieved the highest precision and remained competitive in F1-score.

ORTHRUS yields very few positive detections on Cadets and none on Theia, resulting in low recall and F1 despite high accuracy; this may partly reflect the difference between its attribution-oriented design \cite{jiang2025orthrus} and the detection-focused evaluation used here.

Overall, TGL-APT demonstrates strong detection performance across the three datasets, particularly in recall and F1-score, while maintaining a low false positive rate. These results indicate that TGL-APT can effectively capture diverse APT attack behaviors across different system environments.

\subsection{Efficiency Analysis}
To evaluate the feasibility and runtime efficiency of TGL-APT in practical deployment, we compare it with KAIROS and ORTHRUS in terms of training time, detection latency, and memory consumption. This experiment aims to answer research question (3), i.e., how TGL-APT performs in computational efficiency while maintaining high detection accuracy. The efficiency evaluation considers training time, median and maximum per-window detection runtime, and peak memory usage on the three datasets. The results are reported in Table \ref{tab5.4}.
\begin{table}[htbp]
	\centering
	\renewcommand{\tablename}{Table}
	\caption{Efficiency Comparison of TGL-APT with Other Methods}
	\label{tab5.4}
	\begin{tabular}{c|ccccc}
		\toprule
		Dataset & Method & Training & Median & Max & Memory \\
		&        & Time & Runtime & Runtime & Usage \\
		&        & (min) & (s) & (s) & (GB) \\
		\midrule
		
		\multirow{3}{*}{Cadets}
		& KAIROS  & 44.8  & 2.5    & 19.7   & 8.2 \\
		& ORTHRUS & \textbf{5.903} & 777.24 & 777.24 & \textbf{3.8400} \\
		& TGL-APT & 27.3  & \textbf{1.6} & \textbf{10.6} & 6.4 \\
		\midrule
		
		\multirow{3}{*}{ClearScope}
		& KAIROS  & 12.6  & 1.8   & 14.2  & 5.4 \\
		& ORTHRUS & \textbf{0.569} & 19.72 & 19.72 & \textbf{0.3600} \\
		& TGL-APT & 7.8   & \textbf{1.2} & \textbf{6.7} & 4.1 \\
		\midrule
		
		\multirow{3}{*}{Theia}
		& KAIROS  & 24.2  & 3.2    & 25.4   & 6.7 \\
		& ORTHRUS & \textbf{2.099} & 426.98 & 426.98 & 11.1611 \\
		& TGL-APT & 14.6  & \textbf{2.0} & \textbf{14.5} & \textbf{5.2} \\
		\bottomrule
		
	\end{tabular}
\end{table}

Thanks to the compressed view construction module, which significantly reduces the size of the original provenance graph, TGL-APT achieved an average reduction in training time of approximately 39\% across the three datasets compared to KAIROS. This is mainly because the compressed view filters redundant entities and edges, reducing the amount of graph information processed during temporal learning. The typical detection latency (median time) of TGL-APT was on average about 33\% lower than that of KAIROS, and its maximum per-window detection time was also consistently lower, indicating better response efficiency during routine and high-load detection. The peak memory usage was reduced by approximately 22\% compared to KAIROS, which is consistent with the smaller graph representation maintained during runtime. Although ORTHRUS requires substantially less training time, low online detection latency is particularly important for continuous APT monitoring. TGL-APT therefore trades higher offline training cost for substantially lower detection runtime and stronger detection performance.

To further evaluate the impact of time window scale on memory consumption, we compare TGL-APT with KAIROS using window sizes from 5 to 30 minutes, as shown in Figure~\ref{fig:scalability_window}. As the time window expands, TGL-APT exhibits more gradual memory growth, whereas KAIROS shows a more pronounced increase under full-graph modeling. These results indicate that TGL-APT scales more favorably with increasing time window size.
\begin{figure}[htbp]
	\centering
	\includegraphics[width=0.95\linewidth]{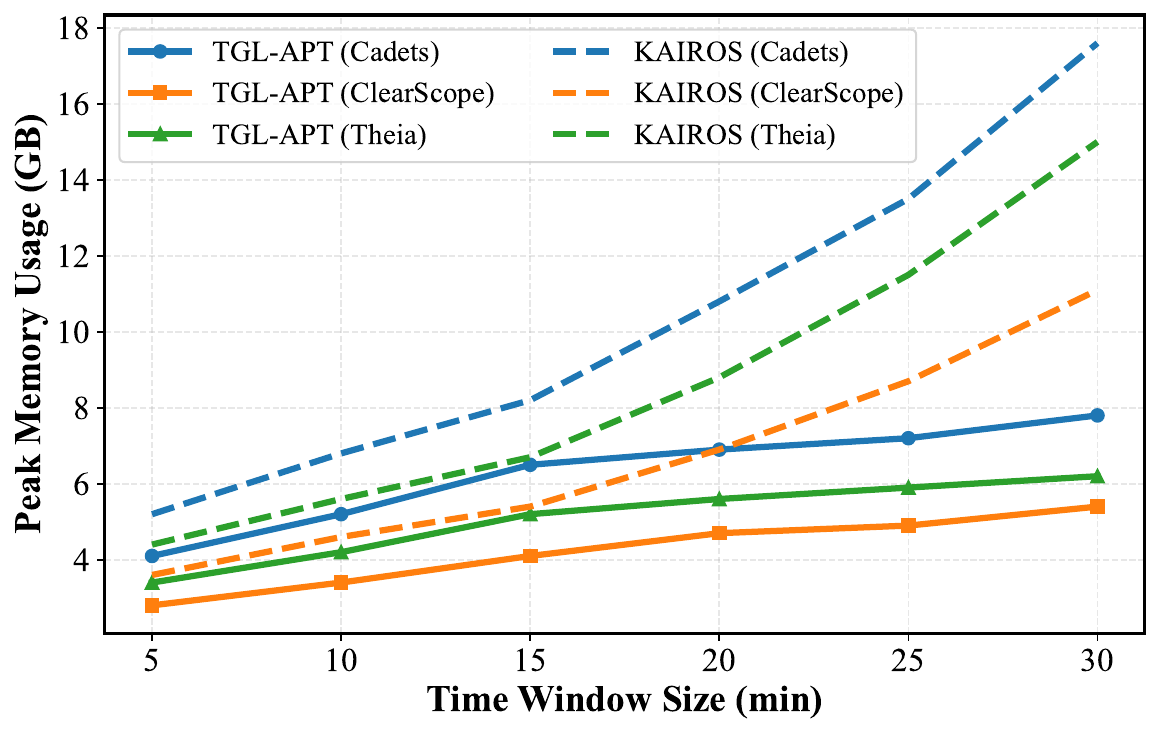}
	\renewcommand{\figurename}{Figure}
	\caption{Peak memory usage with respect to time window scale}
	\label{fig:scalability_window}
\end{figure}

\subsection{Hyperparameter Analysis}
We conduct ablation experiments on four key hyperparameters, including node embedding dimension $|\Phi|$, neighborhood sampling size $|\mathcal{N}|$, edge embedding dimension $|z|$, and time window length $|tw|$, on the Cadets dataset to analyze their impacts on detection performance.

(1) Node Embedding Dimension. The node embedding dimension affects the semantic representation of system entities. As shown in Figure \ref{fig5.3(a)}, performance remains high in a low-dimensional range but degrades gradually when $|\Phi|$ exceeds 64. A moderate dimension is sufficient for effective encoding, while an overlarge dimension introduces sparsity and extra overhead.

(2) Neighborhood Sampling Size. The neighborhood sampling size governs the range of local structural information. The performance is poor at $|\mathcal{N}|=5$, improves significantly at $|\mathcal{N}|=20$, and then saturates. An excessively small size leads to insufficient structural modeling, while an overlarge size brings redundant information without performance gains.

(3) Edge Embedding Dimension. The edge embedding dimension fuses temporal and structural information. As shown in Figure \ref{fig5.3(c)}, the performance increases with dimension and peaks at 200. Excessively high dimensions lead to overfitting and higher computation costs, resulting in performance degradation.

(4) Time Window Length. The time window length controls the temporal granularity. As shown in Figure \ref{fig5.3(d)}, an overly short window fails to capture temporal context, while an overly long window introduces excessive benign interference and reduces recall ( $|tw|=60$). The moderate length $|tw|=15$ achieves the best balance.
\begin{figure}
	\centering
	\subfloat[$|\Phi|$]{
		\label{fig5.3(a)}
		\includegraphics[width=0.47\linewidth]{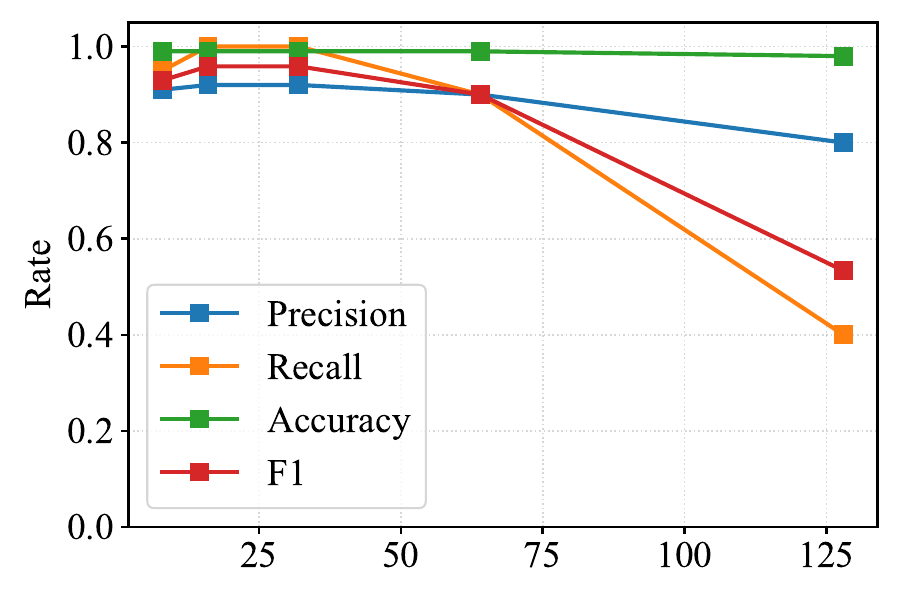}}
	\subfloat[$|\mathcal{N}|$]{
		\label{fig5.3(b)}
		\includegraphics[width=0.47\linewidth]{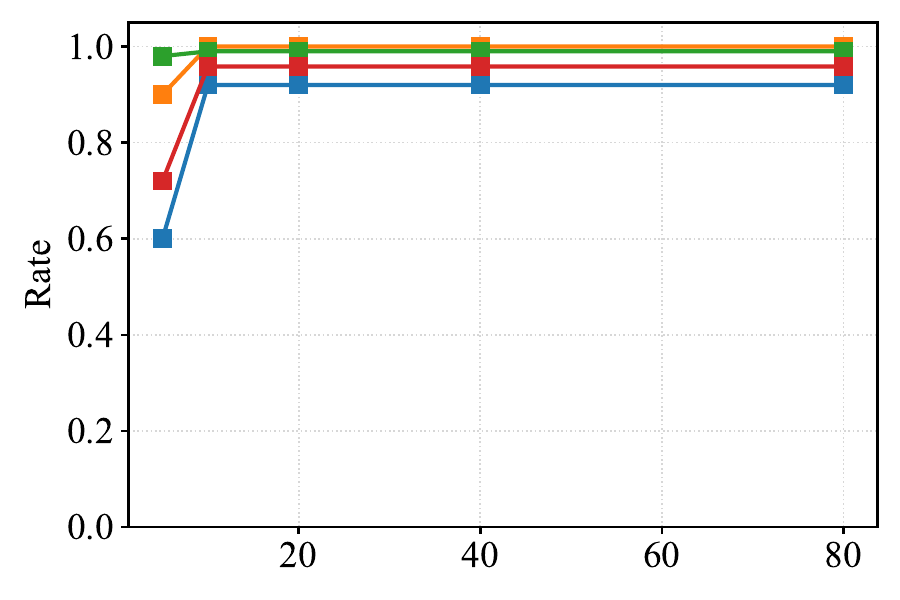}}
	
	\subfloat[$|z|$]{
		\label{fig5.3(c)}
		\includegraphics[width=0.47\linewidth]{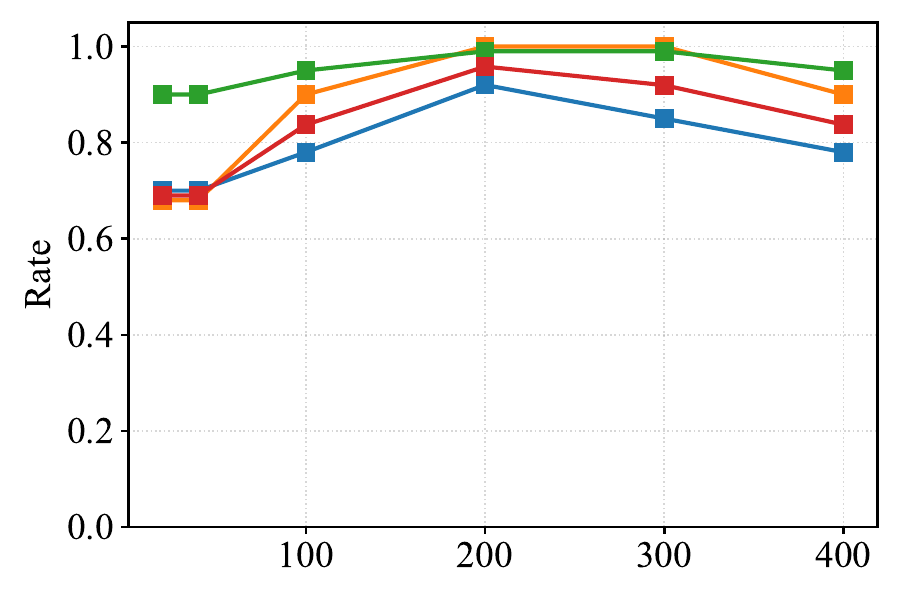}}
	\subfloat[$|tw|$]{
		\label{fig5.3(d)}
		\includegraphics[width=0.47\linewidth]{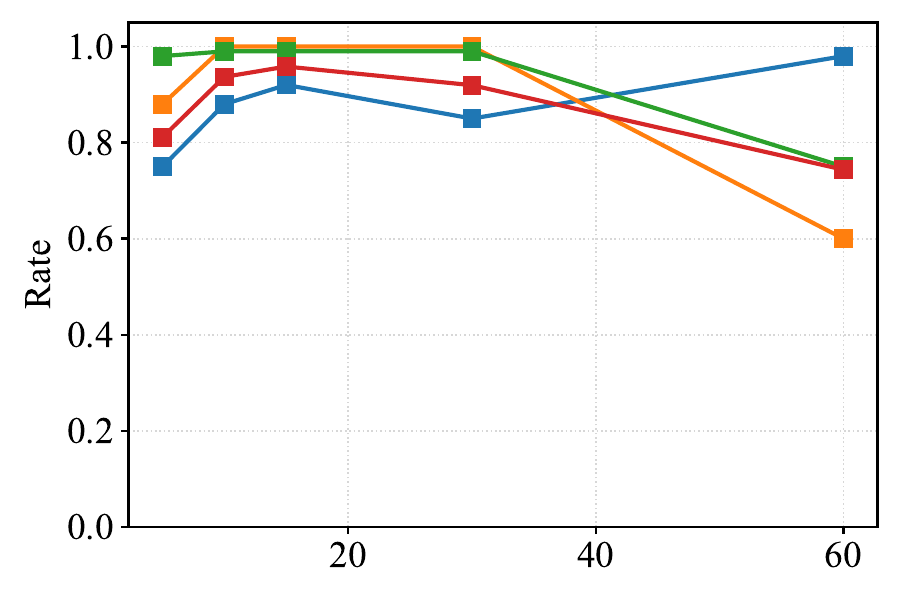}}
	\vspace{8pt}
	\renewcommand{\figurename}{Figure}
	\caption{Impact of different hyperparameters on TGL-APT performance}
	\label{fig5.3}
\end{figure}

\subsection{Ablation Study}
To evaluate the contribution of the key components in TGL-APT, we conduct ablation experiments by removing information-bottleneck graph distillation (w/o IB Distillation), dynamic core node update (w/o DU, Algorithm~\ref{alg5.1}), TF-IDF weighting (w/o TF-IDF), and contrastive learning (w/o Contrastive). These ablations evaluate the effects of provenance distillation, adaptive relevance refinement, and cross-window semantic correlation on detection performance.

The effect of information-bottleneck graph distillation is evaluated on the CADETS dataset. For w/o IB Distillation, core-node filtering and updating are disabled, and temporal learning is performed on all edges of the preprocessed temporal graph. As shown in Table~\ref{tab:ib_ablation}, removing the distillation module causes the F1-score to decrease from 0.9565 to 0.5000, while recall drops from 1.0000 to 0.3636. In particular, the number of true positives decreases from 11 to 4 and false negatives increase from 0 to 7, whereas the number of false positives remains unchanged. This result suggests that directly learning from the undistilled provenance graph makes sparse attack-related activities more difficult to distinguish from large volumes of routine system behavior. By prioritizing structurally and behaviorally informative entities while preserving their causal context, graph distillation helps retain attack-relevant evidence for subsequent temporal learning.
\begin{table}[t]
	\centering
	\caption{Ablation Results of Information-Bottleneck Graph Distillation on CADETS}
	\label{tab:ib_ablation}
	\setlength{\tabcolsep}{2.5pt}
	\begin{tabular}{l|cccc|cccc}
		\toprule
		Variant & TP & TN & FP & FN
		& Precision & Recall & Accuracy & F1 \\
		\midrule
		TGL-APT
		& 11 & 328 & 1 & 0
		& 0.9167 & 1.0000 & 0.9971 & 0.9565 \\
		
		w/o IB Distillation
		& 4 & 328 & 1 & 7
		& 0.8000 & 0.3636 & 0.9765 & 0.5000 \\
		\bottomrule
	\end{tabular}
\end{table}

The dynamic core node update is evaluated across all three datasets. As shown in Figure~\ref{fig5.4a}, removing this mechanism (TGL-APT w/o DU) reduces the F1-score by an average of about 4.5\%. Precision decreases noticeably, while recall also declines slightly on Cadets. This indicates that keeping the initially selected core node set fixed becomes less effective as system behavior evolves. The dynamic update incorporates nodes with high attention weights and large embedding deviations while removing low-contribution nodes, allowing the model to maintain its focus on informative entities over time.

Figure \ref{fig5.4b} shows the comparison results of different correlation strategies. The strategy based solely on entity co-occurrence (w/o TF-IDF) resulted in a sharp drop in recall on the Cadets dataset to 0.64, with the F1-score reduced by 22.0\% compared to the full model. This indicates that related suspicious activities in this scenario frequently involve different entities across time windows, resulting in limited direct overlap of anomalous nodes; entity co-occurrence alone is therefore insufficient to associate fragmented suspicious activities. On ClearScope, although w/o TF-IDF maintained a recall of 1.000, its precision was lower, suggesting that entity co-occurrence introduced a significant number of false associations by incorrectly linking unrelated windows. On the Theia dataset, the performance of w/o TF-IDF was comparable to the full model, presumably because suspicious activities in this dataset exhibited stronger entity co-occurrence, making co-occurrence itself sufficient for window correlation.
\begin{figure}
	\centering
	\subfloat[Dynamic Core Node Update]{
		\label{fig5.4a}
		\includegraphics[width=0.47\linewidth]{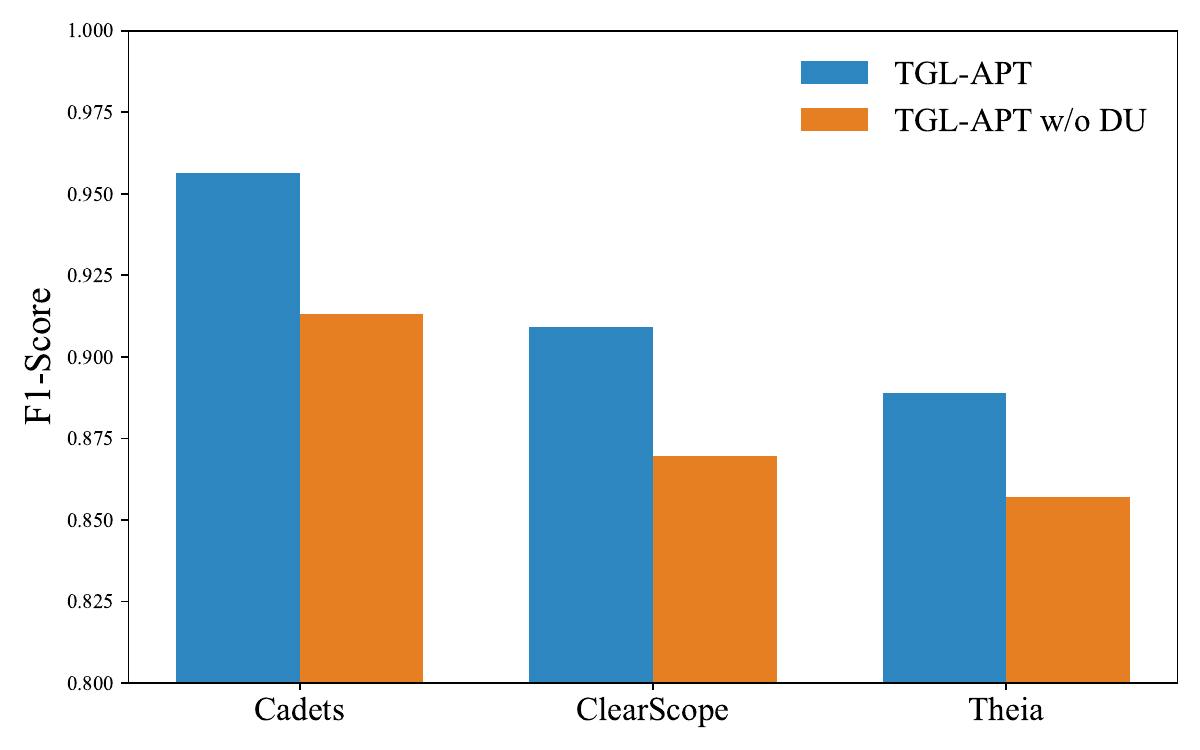}}
	\hfill
	\subfloat[TF-IDF and Contrastive Learning]{
		\label{fig5.4b}
		\includegraphics[width=0.47\linewidth]{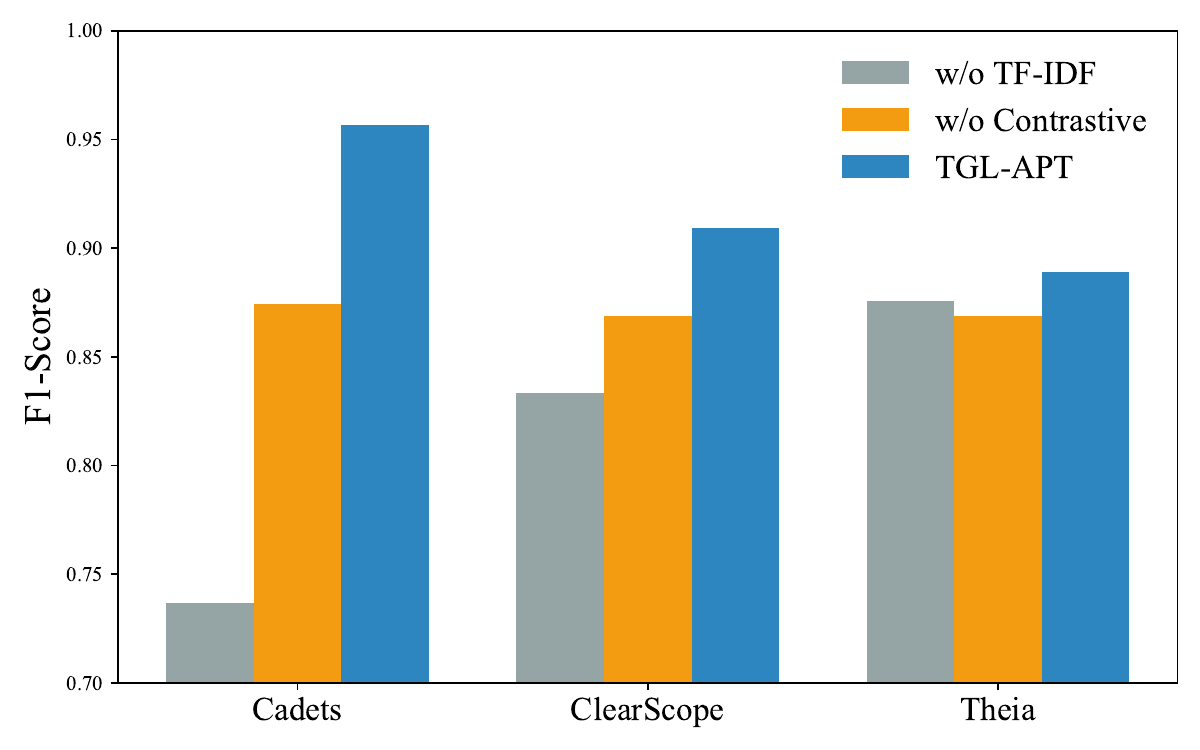}}
	\vspace{8pt}
	\renewcommand{\figurename}{Figure}
	\caption{Ablation study results of key modules}
	\label{fig5.4}
\end{figure}

After retaining TF-IDF weighting but removing contrastive learning (w/o Contrastive), the model's F1-score dropped by 8.6\% and 4.4\% on Cadets and ClearScope, respectively, and by 3.5\% on Theia. This shows that although TF-IDF provides semantic vector representations for suspicious windows, directly computing cosine similarity in the high-dimensional sparse space is less effective for capturing semantic relationships between them. Contrastive learning improves the semantic representation of suspicious windows in a low-dimensional embedding space, showing greater benefits in scenarios with diverse behavioral patterns and frequent entity changes.

\subsection{Attack Investigation Evaluation}
To evaluate the stage categorization component of attack investigation, we conduct a quantitative evaluation on the DARPA E3 dataset. Following the four-stage organization described in Section~\ref{sec:methodology}, the reconstructed attack chain is divided into four stages. Using manually annotated stage labels as ground truth, we compare the stage assignment of each system event (edge) with the algorithm output and report Precision, Recall, and F1-score for each stage, as shown in Figure \ref{fig5.5}. The following case study further illustrates the reconstructed attack chain and its stage organization.
\begin{figure}[hbtp]
	\centering
	\includegraphics[width=0.95\linewidth]{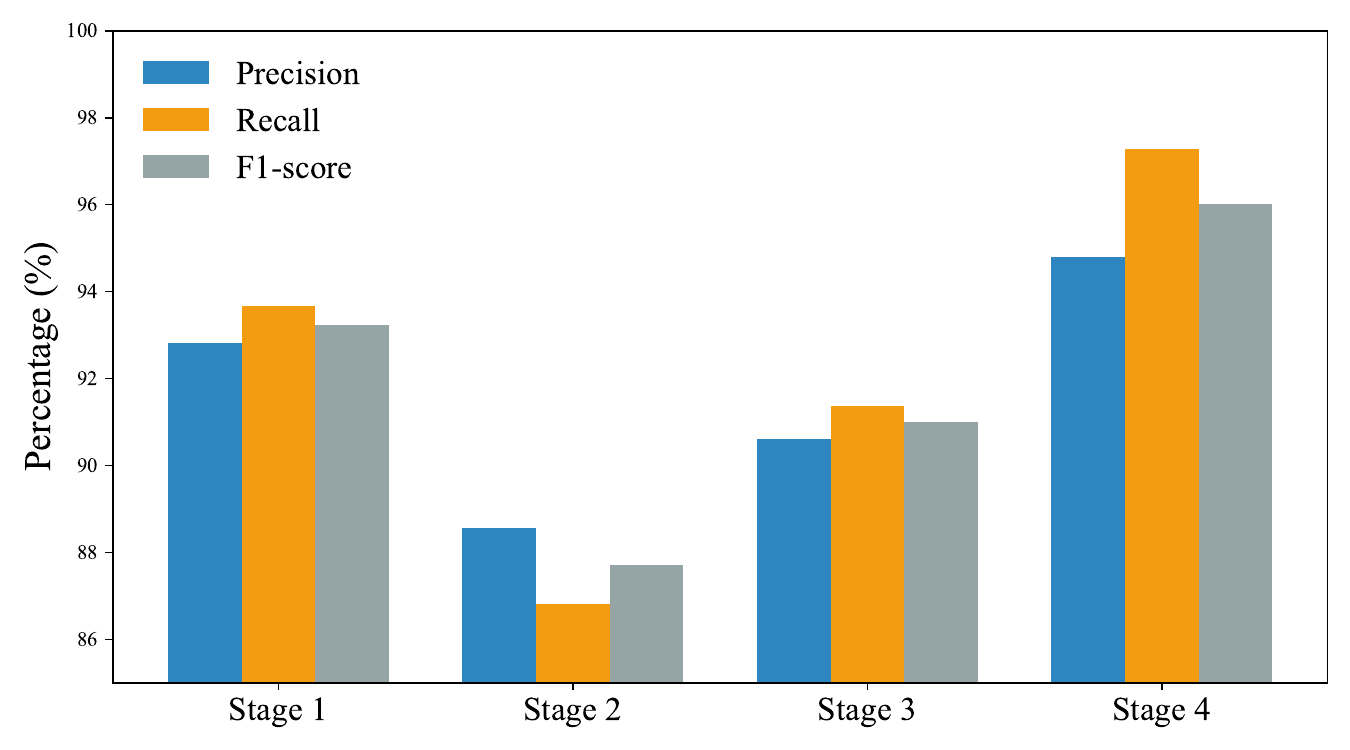}
	\renewcommand{\figurename}{Figure}
	\caption{Evaluation results of attack investigation phase categorization}
	\label{fig5.5}
\end{figure}

As seen in Figure \ref{fig5.5}, the phase categorization algorithm achieved high F1-scores across all four stages, with Stage 4 achieving the highest F1-score, with Precision and Recall reaching 94.79\% and 97.28\%, respectively. This indicates that the attacker's behavior during the final objective achievement and trace cleanup stage is relatively concentrated and feature-rich, allowing the algorithm to identify it accurately. Stage 2 had a relatively lower F1-score, mainly due to a slightly lower recall. Analysis suggests that this stage contains diverse tactical behaviors such as persistence, privilege escalation, and information gathering, which are often mixed with normal system operations, increasing the difficulty of categorization. Nevertheless, its F1-score remained above 87\%, showing that the SCN-based categorization can distinguish the four attack stages under relatively complex behavioral patterns.
\begin{figure*}[t]
	\centering
	\includegraphics[width=1\linewidth]{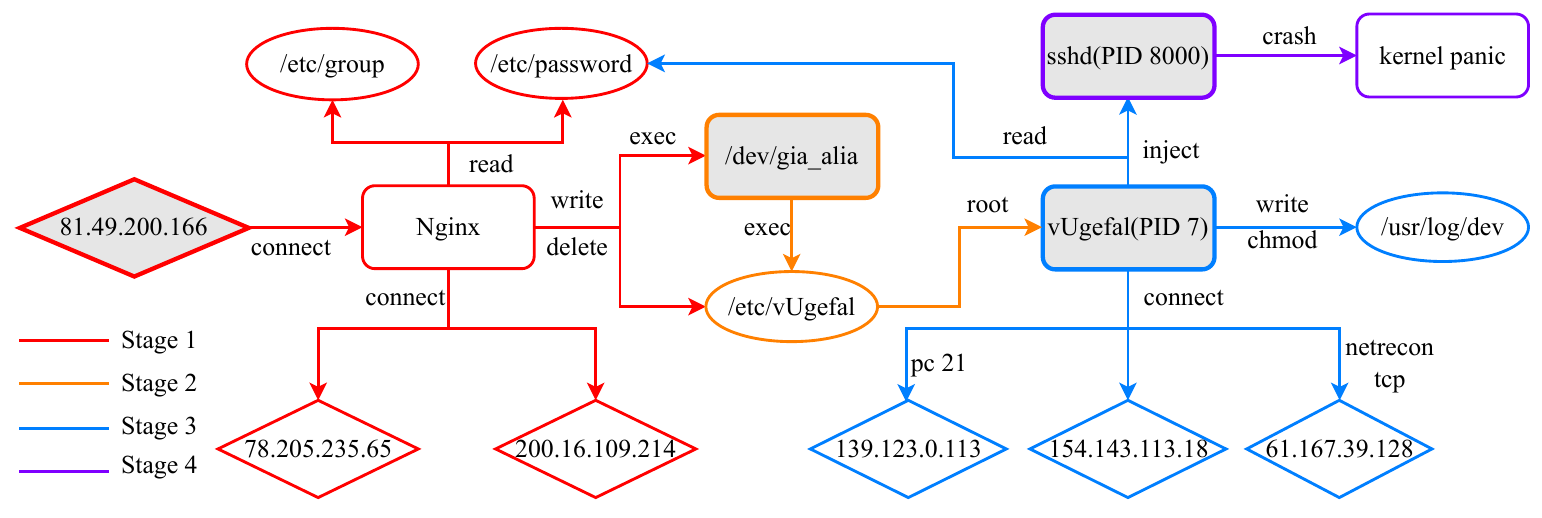}
	\renewcommand{\figurename}{Figure}
	\caption{Example of attack investigation phase categorization on Cadets dataset}
	\label{fig5.6}
\end{figure*}

To complement the quantitative stage-level evaluation, Figure~\ref{fig5.6} presents a qualitative investigation case study on the Cadets dataset, illustrating the reconstructed causal backbone and its subsequent stage categorization. As shown, four stage characterization nodes divide the attack behavior into four stages, namely Stage 1 Initial Intrusion, Stage 2 Residence and Exploration, Stage 3 Lateral Movement, and Stage 4 Objective Achievement and Cover-up. From the IB perspective, these nodes are important not merely because of their local connectivity, but because they mediate transitions between causally distinct attack behaviors. The execution of \texttt{/etc/vUgefal} exhibits high behavioral specificity $R(v)$ because this rare interaction deviates from repetitive benign activities and marks the transition toward persistent malicious behavior. In contrast, the elevated \texttt{sshd} process exhibits strong structural mediation $C(v)$ by bridging local host compromise with multiple external network interactions. These roles are consistent with the causal-bridging and temporal-transition semantics of IB nodes, explaining why these SCNs become representative boundaries between attack stages. Ultimately, the attack triggered a kernel crash to achieve destructive objectives, accompanied by cover-up actions such as log tampering and file deletion. This case study illustrates how TGL-APT organizes reconstructed suspicious activities into an interpretable attack process, providing analysts with a concise view of attack evolution.

\section{Conclusion}
\label{sec:conclusion}

This paper presents TGL-APT, an efficient and adaptive framework for provenance-based APT detection and investigation. TGL-APT combines information-bottleneck-guided graph distillation, adaptive temporal learning, cross-spatiotemporal fingerprint alignment, and causal reconstruction to reduce provenance redundancy and organize fragmented suspicious activities into interpretable attack processes. Experiments on the DARPA E3 datasets demonstrate favorable detection accuracy, computational efficiency, and stage-level investigation capability. Since IB nodes are characterized by information relevance, causal mediation, and temporal behavioral change rather than specific entities or attack patterns, the abstraction also supports generalization across heterogeneous attacks.

\textbf{Limitations and future work.}
Attack-chain reconstruction is currently evaluated through stage-level metrics and case studies, while the adaptive focus mechanism uses a fixed update schedule. Future work will quantitatively evaluate causal-path reconstruction against ground-truth attack graphs and explore adaptive online updates, weakly supervised IB refinement, and structured investigation reports.


% \section*{Acknowledgment}




%\bibliographystyle{IEEEtran}
\bibliographystyle{IEEEtran}
\bibliography{reference}

\end{document}